\documentclass[a4paper,UKenglish,cleveref, autoref, thm-restate]{lipics-v2021}
\hideLIPIcs

\title{A General Technique for Searching in Implicit Sets via Function Inversion}

\author{Boris Aronov}{Tandon School of Engineering, New York University, Brooklyn, NY, USA.}
{boris.aronov@nyu.edu}
{https://orcid.org/0000-0003-3110-4702}
{Work partially supported by NSF Grant CCF-20-08551.}

\author{Jean Cardinal}{Universit\'e libre de Bruxelles (ULB), Brussels, Belgium}{jean.cardinal@ulb.ac.be}{https://orcid.org/0000-0002-2312-0967}{}

\author{Justin Dallant}{Technische Universität Dresden, Dresden, Germany}{justin.dallant@tu-dresden.de}{https://orcid.org/0000-0001-5539-9037}{This work was done while the author was at the Universit\'e libre de Bruxelles, supported by the French Community of Belgium via the funding of a FRIA grant.}

\author{John Iacono}{Universit\'e libre de Bruxelles (ULB), Brussels, Belgium}{john.iacono@ulb.ac.be}{https://orcid.org/0000-0001-8885-8172}{This work was supported by the Fonds de la Recherche Scientifique-FNRS under Grant n°~40008322.}

\authorrunning{B.\ Aronov, J.\ Cardinal,  J.\ Dallant, and J.\ Iacono} 

\Copyright{Boris Aronov, Jean Cardinal,  Justin Dallant, and John Iacono} 



\nolinenumbers

\usepackage{mathtools}

\DeclarePairedDelimiter\floor{\lfloor}{\rfloor}

\usepackage{mdframed}  

\newtheorem{problem}{Problem}

\EventEditors{John Q. Open and Joan R. Access}
\EventNoEds{2}
\EventLongTitle{42nd Conference on Very Important Topics (CVIT 2016)}
\EventShortTitle{CVIT 2016}
\EventAcronym{CVIT}
\EventYear{2016}
\EventDate{December 24--27, 2016}
\EventLocation{Little Whinging, United Kingdom}
\EventLogo{}
\SeriesVolume{42}
\ArticleNo{23}

\begin{document}

\maketitle

\begin{abstract} 
    In recent years, the Fiat-Naor function inversion scheme has been used to disprove conjectures in fine-grained complexity theory and design state of the art data structures for a number of combinatorial problems. We pursue this line of research by considering its application to data structures for searching in implicit sets, defined as the image of a function. Given a function $f$ from the set $[N]$ to a $d$-dimensional integer grid, we consider data structures that allow efficient orthogonal range searching queries in the image of $f$, without explicitly storing it.

    We show that, if $f$ is of the form $[N]\to [2^{w}]^d$ for some $w=\mathrm{polylog} (N)$ and is computable in constant time, then, for any $0<\alpha <1$, we can obtain a data structure using $\tilde {O}(N^{1-\alpha / 3})$ space such that, for a given $d$-dimensional axis-aligned box $B$, we can search for some $x\in [N]$ such that $f(x) \in B$ in time $\tilde{O}(N^{\alpha})$. (Here the $\tilde{O}(.)$ notation omits polylogarithmic factors.)

    Using similar techniques, we further obtain
    \begin{itemize}
    \item data structures for range counting and reporting, predecessor, selection, ranking queries, and combinations thereof, on the set $f([N])$,
    \item data structures for preimage size and preimage selection queries for a given value of $f$, and
    \item data structures for selection and ranking queries on geometric quantities computed from tuples of points in $d$-space.
    \end{itemize}
  
    These results unify and generalize previously known results on 3SUM-indexing and string searching, and are widely applicable as a black box to a variety of problems.   
    In particular, we give a data structure for a generalized version of gapped string indexing, and show how to preprocess a set of points on an integer grid in order to efficiently compute (in sublinear time), for points contained in a given axis-aligned box, their Theil-Sen estimator, the $k$th largest area triangle, or the induced hyperplane that is the $k$th furthest from the origin.
\end{abstract}

\maketitle
\newpage

\section{Introduction}

The computational problem of \emph{inverting} a function $f$ can be cast as follows: Given a value $y$ in the codomain of $f$, decide if there exists an $x$ such that $f(x)=y$, and return such an $x$ if one exists. The function inversion problem is naturally a fundamental issue in cryptography, and the best known result for general function inversion is due to Fiat and Naor~\cite{Fiat1999}:

\begin{theorem}[Fiat-Naor \cite{Fiat1999}]
\label{thm:inversion1}
    For any function $f \colon [N]\to [N]$, and any choice of a constant $0<\alpha <1$, there exists a data structure using $\Tilde{O}(N^{1-\alpha/3})$ space which can invert $f$ at any given point in $\Tilde{O}(N^{\alpha})$ time. Moreover, this data structure can be built in $\Tilde{O}(N)$ time with high probability.
\end{theorem}
Here we use $[N]$ to denote $\{0,\dots,N-1\}$, and use the asymptotic $\Tilde{O}$ notation, that omits polylogarithmic factors.
As was noticed previously by different sets of authors \cite{Corrigan-Gibbs2019, Golovnev2020,Kopelowitz2019}, Theorem~\ref{thm:inversion1} can be generalized to the case where the codomain of $f$ might be larger than the domain of $f$, by using $O(\log N)$ hash functions from a family of pairwise independent functions.

\begin{corollary}[\cite{Corrigan-Gibbs2019, Golovnev2020,Kopelowitz2019}]\label{cor:inversion}
    Theorem~\ref{thm:inversion1} holds for any function $f \colon [N]\to [2^{w}]$, with $w=\mathrm{polylog}(N)$.
\end{corollary}

Recently, the above data structures for function inversion have been applied to tackle fundamental problems that do not directly stem from cryptography: 

\begin{description}
\item[3SUM-indexing]\cite{Kopelowitz2019,Golovnev2020}. Preprocess two sets $S_1$ and $S_2$ of $n$ integers each, so that for any integer $y$, one can decide whether there exist $x_1\in S_1$ and $x_2\in S_2$ such that $y = x_1 + x_2$. 
\item[Collinearity indexing with queries on a line]\cite{AronovESZ23}. Preprocess two sets $S_1$ and $S_2$ of $n$ points in the plane, together with a line $\ell$, so that for any point $y\in\ell$, one can decide whether there exist $p_1\in S_1$ and $p_2\in S_2$ such that $y$, $p_1$, and $p_2$ are collinear.
\item[Gapped string searching]\cite{DBLP:conf/stacs/BilleGLPRS24}. Preprocess a string of length $n$ so that given two patterns $P_1$ and $P_2$ and an integer interval $[\delta, \beta]$, one can report all joint occurrences of $P_1$ and $P_2$ separated by a distance in the interval $[\delta, \beta]$. 
\end{description}

It is a striking fact that, for each of these problems, the best known space-time trade-offs are achieved using such a general tool as the Fiat-Naor inversion scheme. The purpose of this contribution is to extend the applicability of this tool. We show that function inversion, used as a black box, allows to not only decide whether there exists some $x$ such that $f(x)=y$, but to actually search, select, and rank values in the image of $f$, as well as to restrict those values to lie in axis-aligned boxes\footnote{Throughout the paper, we use the term ``box'' as short for ``axis-aligned box'': a hyperrectangle with sides parallel to the coordinate axes.} given as query arguments. 
The space-time tradeoffs are essentially the same as those achieved for function inversion: $k$-ary queries on $n$ items can be answered in strongly sublinear time from a data structure using $O(n^{k-\alpha})$ space, for some $\alpha >0$.
Our contribution can therefore be seen as an efficient reduction of various problems to the function inversion problem.
This provides a simple and powerful method to design data structures for \textit{implicit set representations}, as defined by Chazelle~\cite{Chazelle1987}.

The main idea underlying all our results is that the elementary operations used in range searching, essentially navigating a hierarchy of dyadic boxes, can themselves be implemented by inverting another function, with a slightly larger domain and codomain but within essentially the same space and time bounds. 
The algorithms are elementary, provided that the function inversion data structure is available as a black box. It can be implemented using any such data structure, including some that would be suboptimal but easier to implement.
Conversely, any improvement on the function inversion data structure translates mechanically to an improved algorithm.

\subsection{Previous work}

One of the earliest applications of function inversion in a pure algorithmic context is to the 3SUM-indexing problem defined above, a data structure version of the classical 3SUM problem~\cite{GO12,GP18,KLM19}. Achievable space-time tradeoffs for this question were the topic of several recent results. It was conjectured by Goldstein, Kopelowitz, Lewenstein, and Porat~\cite{GoldsteinKLP17} that any data structure for 3SUM-indexing with $O(n^{\alpha})$ query time for some $\alpha < 1$ must use $\Omega (n^2)$ space. 
This conjecture was disproved by Golovnev, Guo, Horel, Park, and Vaikuntanathan~\cite{Golovnev2020}, and Kopelowitz and Porat~\cite{Kopelowitz2019} using Corollary~\ref{cor:inversion}.
In fact, the 3SUM-indexing problem is a prominent example of a problem that is currently best solved by using this machinery. This is surprising, considering the fact that this solution does not rely on any algebraic property of the problem.

Aronov, Ezra, Sharir, and Zigdon~\cite{AronovESZ23} considered the collinearity indexing problem, which is to the collinearity detection problem what 3SUM-indexing is to the 3SUM problem. There, two sets of $n$ points in the plane are given, and we wish to construct a data structure that answers queries of the following form: Given a  point $q$, are there two points, one in each set, that are collinear with $q$?
Using the Fiat-Naor scheme, they showed that subquadratic space and sublinear query time can be achieved for points with integer coordinates, provided that all query points lie on a line, or more generally on a constant-degree algebraic curve.

By combining function inversion with suffix arrays, Bille, G{\o}rtz, Lewenstein, Pissis, Rotenberg, and Steiner~\cite{DBLP:conf/stacs/BilleGLPRS24} gave improved tradeoffs for gapped string searching data structures. There the problems involve searching multiple patterns in a preprocessed string with constraints on their joint location. 

As mentioned above, we generalize and simplify all these results.

Recently, function inversion was also used by McCauley~\cite{MC24} to reduce the space requirements of a number of nearest-neighbor data structures, in particular those using locality-sensitive hashing.

The question of whether the Fiat-Naor tradeoff in Theorem~\ref{thm:inversion1} can be improved is a tantalizing one. It is not known, in particular, whether a space requirement of only $\tilde{O}(N^{1-\alpha})$ can be achieved, which would be best possible~\cite{GT00,DTT10}. In a recent paper, Golovnev, Guo, Peters, and Stephens-Davidowitz~\cite{GGPS23} showed how to improve the data structure size to $\tilde{O}(N^{\frac 32-\alpha})$, which is beneficial in the regime $\alpha > 3/4$. While such an improvement could in principle be applied to our results, it is not completely clear from their paper that the preprocessing and query time remain as given in Theorem~\ref{thm:inversion1}, since they only focus on space requirements and on the number of queries to the inverted function. Their method also requires shared randomness between the preprocessing and the online phases. While in~practice this can be replaced by a hash function, the result holds in full rigor only in the non-uniform setting, requiring the additional storage of $\tilde{O}(N)$ additional strings that do not depend on the function $f$ being inverted (see the discussion in \cite{GGPS23}, Section 1.1). Golovnev et al. also prove that the function inversion problem as formulated here can be reduced, with only polylogarithmic overhead, to its decision-only variant, in which given a value $y$, we only need to decide whether there exists an $x$ such that $f(x)=y$. The technique used there is reminiscent of the ones we use in our results.

Connections between lower bounds for the function inversion problem, circuit lower bounds, and the network coding conjecture are explored in Corrigan-Gibbs and Kogan~\cite{Corrigan-Gibbs2019}, and in Dvo\v{r}\'ak, Kouck\'y, Kr\'al, and Sl\'ivov\'a~\cite{DKKS21}.

Finally, the relation between the existence of improved function inversion data structures and problems in metacomplexity such as the minimum circuit size problem and the time-Bounded Kolmogorov complexity is investigated in  recent papers from Mazor and Pass~\cite{MP24}, and Hirahara, Ilango, and Williams~\cite{HIW24}. By implementing the Fiat-Naor inversion scheme using Boolean circuits, they managed to improve over the brute-force approach to these problems, something which was once conjectured to be impossible.

\subsection{Plan of the paper}

In the next section, we consider range searching and range counting queries in the image of an integer function. In Section~\ref{sec:beyond}, we use these two basic results to provide a collection of data structures that allow to answer predecessor, ranking, and selection queries in boxes, both in the image and domain of an integer function. In Section~\ref{sec:app}, we illustrate the versatility of the framework on a number of elementary problems. These include generalizations of $k$-SUM indexing data structures and string searching, but also new geometric data structures.

A preliminary version of this paper was presented at the 2024 Symposium on Simplicity in Algorithms (SOSA) \cite{AronovCDI24}. We thank the SOSA referees for insightful comments. This revised version includes in particular an updated summary of recent contributions on data structures involving function inversion as well as an extended section on geometric applications.

\section{Range queries}
\label{sec:range}

We state our two main technical results, involving data structures for range searching and range counting queries in the image of a $d$-dimensional function $f$. 

Our algorithms run in the standard word RAM model of computation with $O(\log N)$ word length. We use the notation $\mathrm{polylog} (N)$ to denote $O(\log^c N)$ for some constant $c$. Note that  arithmetic operations on numbers in the range $[2^w]$ with $w=\mathrm{polylog} (N)$ can be carried out in time $\tilde{O}(1)$ on the word RAM.

\subsection{Range searching}

\begin{theorem}[Range searching]
    \label{thm:range}
    Let $f \colon [N]\to [2^{w}]^d$ be a function computable in $\Tilde{O}(1)$ time, where $w=\mathrm{polylog}(N)$ and $d=O(1)$. For any choice of a constant $0<\alpha <1$, we can construct in $\Tilde{O}(N)$ time with high probability a data structure using $\Tilde{O}(N^{1-\alpha/3})$ words of space which supports the following query:
    Given a $d$-dimensional box $B$, test if there exists $x\in [N]$ with $f(x) \in B$, and report such a value if it exists, in $\Tilde{O}(N^\alpha)$ time.
\end{theorem}

\begin{proof}
    Let $g \colon [N]\times [w+1]^d \to [2^{w+1}]^d \times [w+1]^d$ be the function defined as follows:
    \[g(x,i_1,i_2,\ldots,i_d) \coloneqq \left(\floor*{\frac{f_1(x)}{2^{i_1}}}, \floor*{\frac{f_2(x)}{2^{i_2}}}, \ldots, \floor*{\frac{f_d(x)}{2^{i_d}}},i_1,i_2,\ldots,i_d \right).\]
    The domain of this function has size $O(N(w+1)^d) = \Tilde{O}(N)$ and its codomain has size $O(2^{d(w+1)}(w+1)^d) = \Tilde{O}(2^{\mathrm{polylog}(N)})$, so we can apply Corollary~\ref{cor:inversion} to obtain a data structure using $\Tilde{O}(N^{1-\alpha/3})$ words of space which can invert $g$ at any given point in $\Tilde{O}(N^{\alpha})$ time.

We call a $d$-dimensional box \emph{dyadic} if the coordinates defining the limits of the box in each dimension are successive multiples of a power of two. We can specify such a box using the multiples and powers of two in each dimension: 
\begin{align*}
 DB(y,i) = {}&DB(y_1,y_2\ldots,y_d,i_1,i_2,\ldots i_d) \coloneqq \\ 
             &[2^{i_1}y_1,2^{i_1}(y_1+1)-1]\times [2^{i_2}y_2,2^{i_2}(y_2+1)-1] \times \\
  &\quad \dots \times [2^{i_d}y_d,2^{i_d}(y_d+1)-1].
\end{align*}

Our function $g$ has been defined so that 
\[ f(x)  \in  DB(y,i)
\longleftrightarrow
g(x,i)=(y,i).
\]
Thus, given any dyadic box $DB(y,i)$, we can determine if there is an $x$ such that $f(x) \in DB(y,i)$ by asking if $(y,i)$ is in the image of $g$. 

Given a  dyadic box $DB(y,i)$ that contains $f(x)$ for some $x$ that we wish to compute, each dyadic box with $d'\leq d$ non-zero elements of $i$ can be partitioned into $2^{d'}$ dyadic boxes where the non-zero elements of $i$ are decremented by one; call these the child boxes of $DB(y,i)$. These child boxes can equivalently be viewed as the result of cutting $DB(y,i)$ in half in each dimension where its size is at least 2. If some $f(x) \in DB(y,i)$, then $f(x) \in B'$ for some child box $B'$ of $DB(y,i)$; continue this process until some $f(x) \in DB(z,0)$ is found, which means $f(x)=z\in DB(y,i)$. This will involve $w+1$ steps, each consisting of at most $2^d$ inversion queries.

Let $B$ be a $d$-dimensional box, not necessarily dyadic. To determine if there is an $f(x) \in B$, we use the fact that any $d$-dimensional box defined by integers at most $2^{\mathrm{polylog}(N)}$ can be partitioned into $O(\log^d (2^{\mathrm{polylog}(N)})) = \Tilde{O}(1)$ dyadic boxes, and query each of these separately.
\end{proof}

Note that in one dimension the idea of searching in bounded universes via hashing is well known, from, for example, the $x$- and $y$-fast tries of Willard~\cite{DBLP:journals/ipl/Willard83}. The idea of extending this to higher dimensions with dyadic boxes has been used in the past, for example in the point location method of Iacono and Langerman~\cite{DBLP:conf/cccg/IaconoL00}.

\subsection{Range counting}

We now show how to extend our result from searching to counting using a big/small approach whereby counts of dyadic boxes intersecting large number of images of $f$ have their counts stored explicitly, and the remaining dyadic boxes have their counts computed on demand using function inversion. This increases the space usage from 
$\Tilde{O}(N^{1-\alpha/3})$ to
$\Tilde{O}(N^{1-\alpha/4})$ while supporting  $\Tilde{O}(N^\alpha)$-time queries.

\begin{theorem}[Range counting]
\label{thm:counting}
    Let $f \colon[N]\to [2^{w}]^d$ and $c \colon[2^{w}]^d\to [2^{w}]$ be functions computable in $\Tilde{O}(1)$ time, where $w=\mathrm{polylog}(N)$ and $d=O(1)$. For any choice of a constant $0<\alpha <1$, we can construct in $\Tilde{O}(N)$ time with high probability a data structure using $\Tilde{O}(N^{1-\alpha/4})$ words of space which supports the following query:
    Given a $d$-dimensional box $B$, return $\sum_{y\in f([N])\cap B} c(y)$, in $\Tilde{O}(N^\alpha)$ time.
\end{theorem}

\begin{proof} $ $
    Consider again the function $g$ defined in the proof of Theorem~\ref{thm:range}, with the same inversion data structure using $\Tilde{O}(N^{1-\alpha/3})$ words of space which can invert $g$ at any given point in $\Tilde{O}(N^{\alpha})$ time. 

    For every dyadic box $DB(y,i)$ which contains at least $N^{\alpha/3}$ elements of $f([N])$, store that box 
    as a key with value $\sum_{y\in f(N)\cap DB(y,i)} c(y)$ in a dictionary. This dictionary can be implemented in any number of ways, for instance using hash tables or binary search trees. As each element of $f([N])$ appears in at most $O(\log^d(2^{\mathrm{polylog}(N)}))=\Tilde{O}(1)$ dyadic boxes, the number of dyadic boxes with at least $N^{\alpha/3}$ elements is at most $\Tilde{O}(N^{1-\alpha/3})$. Thus, this does not increase space usage significantly.

    Now, suppose we are given a dyadic box $DB(y,i)$ and want to find the number of elements of $f([N])$ in that box. We start by checking if this box contains at least $N^{\alpha/3}$ elements by a lookup in the dictionary, and we can report the stored number if it is the case. Otherwise, the box contains at most $N^{\alpha/3}$ elements and we recurse on all non-empty children of $DB(y,i)$. We keep a global counter $s$ to which we add $c(z)$ each time we reach a non-empty box of the form $DB(z,0)$. Because each element of $f([N])$ appears in at most $\Tilde{O}(1)$ dyadic boxes, this requires $\Tilde{O}(N^{\alpha/3})$ inversion queries, for a total query time of $\Tilde{O}(N^{\alpha/3}\cdot N^{\alpha}) = \Tilde{O}(N^{4\alpha/3})$.

    Given a non-dyadic box, we can again partition it into $O(\log^d (2^{\mathrm{polylog}(N)})) = \Tilde{O}(1)$ dyadic boxes and query each of these separately.

    Finally, rescaling $\alpha$ by a factor of $3/4$ gives the result.
\end{proof}

The previous proof can easily be adapted to support other types of associative ("semigroup") operations instead of a sum, such as taking the maximum in a box.

\section{Beyond range queries}    
\label{sec:beyond}

We now give general results on the existence of data structures for other, more complex types of queries. The proofs are elementary and only rely on Theorems~\ref{thm:range} and \ref{thm:counting}. Note that some of these results can be proved directly using similar ideas as for Theorems~\ref{thm:range} and \ref{thm:counting} and with a slightly better running time if we take polylogarithmic factors into account. 

\subsection{Predecessor, ranking, and selection queries}

We first observe that the range searching data structures described above directly yield data structures for predecessor, ranking, and selection queries.
In the following two lemmas, we let $f \colon[N]\to [2^{w}]$ be a function computable in $\Tilde{O}(1)$ time, where $w=\mathrm{polylog}(N)$.

\begin{theorem}[Predecessor search]
    For any choice of a constant $0<\alpha <1$, we can construct in $\Tilde{O}(N)$ time with high probability a data structure using $\Tilde{O}(N^{1-\alpha/3})$ words of space which supports the following query:
    Given an integer $y\in [2^{w}]$, report $x\in [N]$ such that $f(x)$ is the largest value in $f([N])$ that is smaller or equal to $y$, if it exists, in $\Tilde{O}(N^\alpha)$ time.
\end{theorem}  
\begin{proof}
    We bisect on the positive integers $z\leq y$. For such an integer $z$, we define $B \coloneqq [z,y]$ and use the data structure of Theorem~\ref{thm:range} with $d=1$ to decide whether there exists $x$ such that $f(x)\in B$ in time $\tilde{O}(N^{\alpha})$. The number of such queries is at most $\log (2^{w}) = w = \mathrm{polylog}(N)$, hence the overall cost remains $\tilde{O}(N^{\alpha})$.
\end{proof}

\begin{theorem}[Ranking and Selection]
    For any choice of a constant $0<\alpha <1$, we can construct in $\Tilde{O}(N)$ time with high probability a data structure using $\Tilde{O}(N^{1-\alpha/4})$ words of space which supports the following queries in $\Tilde{O}(N^\alpha)$ time:
    \begin{description}
    \item[Ranking.] Given an integer $y\in [2^{w}]$, return $\left| \strut\{x\in [N] : f(x)<y\}\right|$.
    \item[Selection.] Given an integer $k\in [N]$, report $x\in [N]$ such that $f(x)$ is the $k$th largest value in $f([N])$, if it exists.
    \end{description}
\end{theorem}
\begin{proof}
    For ranking, we let the function $c$ be such that $c(x)=1$ for all $x\in [2^{w}]$, and use the data structure of Theorem~\ref{thm:counting} with $d=1$. Querying the interval $B \coloneqq [y]$ in this structure yields the answer in time $\Tilde{O}(N^\alpha)$.

    For selection, we bisect over the values $y\in [2^{w}]$ and repeatedly query the data structure on the interval $B \coloneqq [y]$. This requires $\mathrm{polylog}(N)$ steps, hence time $\Tilde{O}(N^\alpha)$ overall as well.
\end{proof}

\subsection{Ranking and selection with range constraints}

By applying Theorems~\ref{thm:range} and \ref{thm:counting} with $d>1$, we can also perform predecessor, ranking, and selection queries within a given box.

\begin{theorem}
\label{thm:countinrange}
    Let $f \colon [N]\to [2^{w}]^d$ and $\mu \colon [2^{w}]^d\to [2^{w}]$ be functions computable in $\Tilde{O}(1)$ time, where $w=\mathrm{polylog}(N)$ and $d=O(1)$. For any choice of a constant $0<\alpha <1$, we can construct in $\Tilde{O}(N)$ time with high probability a data structure using $\Tilde{O}(N^{1-\alpha/4})$ words of space which supports the following queries in $\Tilde{O}(N^\alpha)$ time:
    \begin{description}
    \item[Ranking in a box.] Given a $d$-dimensional box $B$ and an integer $y\in [2^{w}]$, return the value
    
    $\left|\strut\{f(x): x\in [N], f(x)\in B, \mu(f(x))<y\}\right|$.
    \item[Selection in a box.] Given a $d$-dimensional box $B$ and an integer $k\in [N]$, report $x\in [N]$ such that $\mu(f(x))$ is the $k$th largest value in $\{\mu(f(x)) : x\in [N], f(x)\in B \}$, if it exists.
    \item[Median in a box.] Given a $d$-dimensional axis aligned box $B$, report $x\in [N]$ such that $\mu(f(x))$ is the median of $\{\mu(f(x)) : x\in [N], f(x)\in B\}$.
    \end{description}
\end{theorem}
\begin{proof}
    To perform ranking of $y$ in a box $B$, we can construct the data structure of Theorem~\ref{thm:counting} on the function $g \colon [N]\to [2^{w}]^{d+1}$ defined by
    \[
    g(x) \coloneqq (f(x), \mu(f(x))),
    \]
    let the count function $c$ for the data structure be equal to 1 everywhere, and query the data structure with the $(d+1)$-dimensional box $B\times [y]$ in time $\Tilde{O}(N^\alpha)$.

    Selection in a box can be done as for ranking, except we now bisect on the values $y$, which only increases the running time by a $\mathrm{polylog}(N)$ factor.

    Finally, we can select the median by first counting the number $m$ of values in the preimage of $B$, then performing selection with $k=m/2$.
\end{proof}

Note that this result can be used to select a value among preimages, as follows.

\begin{theorem}[Preimage queries] \label{lem:preimage}
    Let $f\colon[N]\to [2^{w}]$ be a function computable in $\Tilde{O}(1)$ time, where $w=\mathrm{polylog}(N)$.
    For any choice of a constant $0<\alpha <1$, we can construct in $\Tilde{O}(N)$ time with high probability a data structure using $\Tilde{O}(N^{1-\alpha/4})$ words of space which supports the following queries in $\Tilde{O}(N^\alpha)$ time:
    \begin{description}
    \item[Preimage ranking.] Given two integers $y,z\in [2^{w}]$, return
      \[\left|\strut\{x\in [N] : f(x)=y, x<z\}\right|.\]
    \item[Preimage selection.] Given an integer $k\in [N]$ and an integer $y\in [2^{w}]$, report the $k$th largest value $x\in [N]$ such that $f(x)=y$, if it exists.
    \item[Preimage median.] Given an integer $y\in [2^{w}]$, return the median of the values $x\in [N]$ such that $f(x)=y$, if such a value exists.
    \end{description}
\end{theorem}
\begin{proof}
    Let $g \colon [N]\to [2^{w}]\times[N]$ be the function defined by
    \[
    g(x) \coloneqq (f(x), x).
    \]
     Preimage ranking, selection, and median then reduces directly to ranking, selection, and median in a box respectively, on the function $g$.
     The result then follows from Theorem~\ref{thm:countinrange}.
\end{proof}

\subsection{Range reporting}

One can also consider the problem of \emph{reporting} distinct values $x$ such that $f(x)$ is contained in a specific box. The following can be obtained from Theorem~\ref{thm:range}.


\begin{theorem}[Range reporting]

\label{thm:reportinrange}
    Let $f \colon [N]\to [2^{w}]^d$ be computable in $\Tilde{O}(1)$ time, where $w=\mathrm{polylog}(N)$ and $d=O(1)$. For any choice of $0<\alpha <1$, we can construct in $\Tilde{O}(N)$ time with high probability a data structure using $\Tilde{O}(N^{1-\alpha/3})$ words of space which supports the following query in $\Tilde{O}((k+1)N^\alpha)$ time:
    Given a $d$-dimensional box $B$ and and an integer $k\in [N]$, return $k$ distinct values $x\in [N]$ such that $f(x)\in B$, if they exist. (If only $k'<k$ such values exist, they are all returned in $\Tilde{O}((k'+1)N^\alpha)$ time.)
\end{theorem}
\begin{proof}
    From the data structure of Theorem~\ref{thm:range} constructed on the function $g(x) \coloneqq (f(x), x)$, we can decide in time $\Tilde{O}(N^\alpha)$ whether there exists any $x\in [N]$ such that $f(x)\in B$, and $x<y$ for a given box $B$ and integer $y$. This allows to perform a bisection search on the coordinates of $B$ and $y$ and isolate up to $k$ distinct values $x$ such that $f(x)\in B$, in $\mathrm{polylog}(N)$ time for each.
\end{proof}

Function inversion was recently used by McCauley~\cite{MC24} to improve the space efficiency of approximate nearest neighbor data structures. As a building block towards this goal, the following was shown.

\begin{theorem}[All-Function Inversion Data Structure \cite{MC24}]
    Let $f_0,f_1\ldots,f_{R-1} \colon [N]\to [2^{w}]$ be $R$ functions each computable in $\Tilde{O}(1)$ time, where $w=\mathrm{polylog}(N)$ and $d=O(1)$. For any choice of $0<\alpha <1$, we can construct in $\Tilde{O}(RN)$ time with high probability a data structure using $\Tilde{O}(N+RN^{1-\alpha/3})$ words of space which for any query $(i,y) \in [R]\times [2^{w}] $ returns $f_i^{-1}(y)$ in $\Tilde{O}((|f_i^{-1}(y)|+1)N^\alpha)$ time.
\end{theorem}

Note that this result can now be obtained by simply applying Theorem \ref{thm:reportinrange} to each $f_i$, and querying the relevant data structure with $k=N$ and $B=\{y\}$. In fact, this even saves on the $\Tilde{O}(N)$ term in the space used.

\section{Applications}
\label{sec:app}

We can apply our general method not only to new variants of problems for which function inversion was already known to be useful~\cite{GoldsteinKLP17,Golovnev2020,Kopelowitz2019,DBLP:conf/stacs/BilleGLPRS24,AronovESZ23}, but also to indexing variants of a number of other previously studied problems in computational geometry~\cite{Chazelle1987,DBLP:journals/siamcomp/ColeSSS89,DBLP:journals/algorithmica/AgarwalASS93,Bespamyatnikh2004,BarbaCILOS19,CardinalS23}.

\subsection{Variants of 3SUM-indexing}

Recall that the 3SUM-indexing problem consists of preprocessing two sets $S_1$ and $S_2$, each of $n$ integers, so that for any integer $y$, one can decide efficiently whether there exist $x_1\in S_1$ and $x_2\in S_2$ such that $y=x_1+x_2$~\cite{GoldsteinKLP17,Golovnev2020,Kopelowitz2019}.

As an example of a simple generalization that we can tackle using our results, we consider the following generalization of the $k$-SUM-indexing problem to arbitrary constant-degree polynomial functions, yielding a data structure version of the recently studied 3POL and $k$-POL problems~\cite{BarbaCILOS19,CardinalS23}. 

\begin{problem}[$k$-POL-indexing problem]
    Let $S_1,S_2,\ldots ,S_k$ be $k$ sets of $n$ numbers in $[2^{w}]$, where $k=O(1)$ and $w=\mathrm{polylog}(n)$, and let $p\in \mathbb{N}[x_1,x_2,\ldots ,x_k]$ be a constant-degree $k$-variate polynomial with coefficients in $[2^{w}]$.
    Preprocess these sets and $p$ to answer queries of the following form: Given $k$ intervals $B_1,B_2,\ldots ,B_k\subset [2^{w}]$ and an integer $y$, are there $k$ values $x_1\in B_1\cap S_1, x_2\in B_2\cap S_2,\ldots ,x_k\in B_k\cap S_k$ such that $p(x_1,x_2,\ldots ,x_k)=y$?  
\end{problem}

\begin{theorem}[$k$-POL-indexing problem]
    For any choice of a constant $0<\alpha <1$, we can construct in $\Tilde{O}(n^k)$ time with high probability a data structure for the $k$-POL-indexing problem using $\Tilde{O}(n^{k-\alpha/3})$ words of space which supports queries in $\Tilde{O}(n^\alpha)$ time.
\end{theorem}
\begin{proof}
    We can assume without loss of generality that the values of the polynomial $p$ lie in $[2^{w'}]$ for some $w'\geq w$ such that $w'=\mathrm{polylog} (n)$.
    We define the function $f\colon[n^k]\to [2^{w'}]^{k+1}$ that maps a $k$-tuple $(i_1,i_2,\ldots ,i_k)$ to the corresponding $(k+1)$-tuple of integers $(x_1,x_2,\ldots ,x_k,p(x_1,x_2,\ldots ,x_k))$, where $x_j$ is the $i_j$th largest value in $S_j$. We then construct the data structure of Theorem~\ref{thm:range} with $d=k+1$ and $N=n^k$, and use it to answer queries with the $(k+1)$-dimensional boxes $B \coloneqq B_1\times B_2\times\ldots\times B_k\times \{y\}$.
\end{proof}

Note that from Theorem~\ref{thm:countinrange}, with a slight penalty in space, we can also answer ranking and selection queries in the preimage $p^{-1}(y)\cap B$ in sublinear time. 

\subsection{Substring search}

Function inversion has been applied to string searching problems before~\cite{Corrigan-Gibbs2019,DBLP:conf/stacs/BilleGLPRS24}.
Let us first consider the following problem:

\begin{problem}[Birange proximity]
    Given an array $A$ of size $n$ containing elements of $[n]$. Preprocess it to answer queries of the form: Given intervals $[i,j]$, $[k,l]$, and $[\delta,\beta]$ all with endpoints in $[n]$,  report all pairs $(x,y)$, $x \in [i,j]$ and $y \in [k, l]$, such that $|A[x]-A[y]| \in [\delta,\beta]$.
\end{problem}
We can define $f\colon[n]^2  \to [n]^3$ by $f(x,y)\coloneqq (x,y,|A[x]-A[y]|)$.
Answering a birange proximity query is thus equivalent to determining the pairs $(x,y)$ with $f(x,y)=(x,y,\gamma)$, $x \in [i,j]$, $y\in [k,\ell]$, and $\gamma \in [\delta,\beta]$. 
Thus, by applying Theorem~\ref{thm:reportinrange} we immediately obtain:

\begin{lemma}[Birange proximity]
    For any choice of a constant $0<\alpha <1$, we can construct in $\Tilde{O}(n^2)$ time with high probability a data structure for the birange proximity problem using $\Tilde{O}(n^{2-\alpha/3})$ words of space which supports queries that return $k$ pairs in time $\Tilde{O}((k+1)n^\alpha)$.
\end{lemma}

Bille et al.~\cite{DBLP:conf/stacs/BilleGLPRS24} consider the following problem:
\begin{problem}[Gapped string indexing]
    Given a string $S$ of length $n$, prepossess it to support queries of the form: Given two strings $P_1$ and $P_2$, and an interval $[\delta,\beta]$, report all $(x,y)$ such that $P_1$ is a substring starting at location $x$ in $S$, $P_2$ is a substring starting at location $y$ in $S$, and $|x-y| \in [\delta,\beta]$.
\end{problem}

The suffix array $A$ of a string $S$ is a fundamental data structure in stringology due to Manber and Meyers~\cite{DBLP:journals/siamcomp/ManberM93}. It has $A[j]=i$ if $S[i\ldots n]$ is the $j$th lexicographically largest suffix of $S$, and it can be computed in linear time~\cite{DBLP:journals/jacm/Farach-ColtonFM00,DBLP:journals/jda/KoA05, DBLP:journals/jda/KimSPP05}.
Bille et al.~\cite{DBLP:conf/stacs/BilleGLPRS24} observe that given the suffix array $A$ of $S$, which is a permutation of $[n]$, the locations of all occurrences of a substring $P$ are the values in a consecutive interval of $A$, and this interval can be computed in time $O(|P|)$ using a suffix tree (the compressed trie of suffixes, which can also be computed in linear time~\cite{F97,DBLP:journals/jacm/Farach-ColtonFM00}). Using a birange proximity structure, initialized with the suffix array of $S$, a gapped string indexing query in $S$ can be converted into a birange proximity query at cost $O(|P_1|+|P_2|)$. We therefore immediately recover the headline result of~\cite{DBLP:conf/stacs/BilleGLPRS24}:

\begin{theorem}[Gapped string indexing]
        For any choice of a constant $0<\alpha <1$, we can construct in $\Tilde{O}(n^2)$ time with high probability a data structure for the gapped string indexing problem using $\Tilde{O}(n^{2-\alpha/3})$ words of space which supports queries that return $k$ pairs in time $\Tilde{O}(|P_1|+|P_2|+(k+1)n^\alpha)$.
\end{theorem}

While the problem we call birange proximity is implicit in the presentation of Bille et al.~\cite{DBLP:conf/stacs/BilleGLPRS24}, their reductions make use of other problems; here the use of birange proximity along with the framework we have developed makes the proof immediate and amenable to many possible variations. We state one possible generalization, where instead of two strings, $d$~substrings are to be matched, with $d\ge2$ a constant.  We impose pairwise restrictions on the allowable distances between the matches, and additionally we restrict the occurrence of each substring to an interval in the base string.

\begin{problem}[Generalized gapped string indexing]
    Given a string $S$ of length $n$ and a constant nonnegative integer $d$, preprocess to support queries of the form: Given strings $P_0, P_1, \ldots P_{d-1}$, 
    intervals $[\delta_i,\beta_j]$ for all $i,j\in [d]$, $i<j$, 
intervals $[\gamma_i,\xi_i]$, for all $i \in [d]$    
    report all $(x_0, x_1, \ldots x_{d-1})$ such that $P_i$ is a substring starting at location $x_i$ in $S$, $x_i \in [\gamma_i,\xi_i]$ and $|x_i-x_j| \in [\delta_i,\beta_j]$ for all $i,j\in [d]$, $i<j$. The counting variant returns the number of $(x_0, x_1, \ldots , x_{d-1})$ meeting this condition.
\end{problem}

Using Theorems~\ref{thm:countinrange} and ~\ref{thm:reportinrange} with $N=\Tilde{O}(n^k)$, we directly obtain the following.

\begin{theorem}[Generalized gapped string indexing]
        For any choice of a constant $0<\alpha <1$, we can construct in $\Tilde{O}(n^d)$ time with high probability a data structure for the generalized gapped string indexing problem using $\Tilde{O}(n^{d-\alpha/3})$ words of space which supports queries that return $k$ pairs in time $\Tilde{O}(\sum_{i=1}^d |P_i|+(k+1)n^\alpha)$.
        For the counting variant, queries take time $\Tilde{O}(\sum_{i=1}^d |P_i|+n^\alpha)$ and the space usage is $\Tilde{O}(n^{d-\alpha/4})$.
\end{theorem}

\subsection{Geometric applications}

The following result follows directly from Theorem~\ref{thm:countinrange} and will be useful in applying our method to specific computational geometry problems.

\begin{lemma}[Searching tuples of points in space]
\label{lem:points}
Let $S$ be a set of $n$ points on the $d$-dimensional integer grid $[2^{w}]^d$, where $w=\mathrm{polylog}(n)$ and $d=O(1)$.
Let $\delta :[2^w]^{td}\to [2^{w}]$, for some constant $t>1$, be any integer function mapping a $t$-tuple of points in $S$ to an integer, and computable in $\tilde{O}(1)$. 

    For any choice of a constant $0<\alpha <1$, we can construct in $\Tilde{O}(n^t)$ time with high probability a data structure using $\Tilde{O}(n^{t-\alpha/4})$ words of space which supports the following queries in $\Tilde{O}(n^\alpha)$ time:
    \begin{description}
    \item[Ranking in a box.] Given $t$ $d$-dimensional boxes $B_1,B_2,\ldots,B_t$ and an integer $y$, return the number of $t$-tuples $T$ of distinct points of $S$ in $B_1\times B_2\times\cdots\times B_t$ such that $\delta (T)<y$. 
    \item[Selection in a box.] Given $t$ $d$-dimensional boxes $B_1,B_2,\ldots,B_t$ and an integer $k$, return the $k$th largest value of $\delta (T)$ among all $t$-tuples $T$ of distinct points of $S$ in $B_1\times B_2\times\cdots\times B_t$.
    \item[Median in a box] Given $t$ $d$-dimensional boxes $B_1,B_2,\ldots,B_t$, return the median of $\delta (T)$ among all $t$-tuples $T$ of distinct points of $S$ in $B_1\times B_2\times\cdots\times B_t$.
    \end{description}
\end{lemma}
\begin{proof}
We denote by $p_{i,\ell}$ the $\ell$th coordinate of the $i$th point of $S$, for $i\in [n]$.
Consider the function $g$ mapping an ordered $t$-tuple of pairwise distinct elements of $[n]$ to $[2^{w}]^{td}$, defined as
\[
    g((i_1,i_2,\ldots ,i_t)) := (
    p_{i_1,1}, p_{i_1,2}, \ldots ,p_{i_1,d},
    p_{i_2,1}, p_{i_2,2}, \ldots ,p_{i_2,d},\ldots ,
    p_{i_t,1}, p_{i_t,2}, \ldots ,p_{i_t,d}).
\]
The result follows by applying Theorem~\ref{thm:countinrange} to the functions $g$ and $\delta$, with $N=n!/(n-t)!<n^t$.
\end{proof}

One simple example of a function $\delta$ for $t=2$ is the squared Euclidean distance between two points, and we then get a data structure for distance selection and ranking. 
The following result provides a data structure version of the distance selection problem~\cite{DBLP:journals/algorithmica/AgarwalASS93,DBLP:conf/compgeom/Goodrich93,DBLP:journals/siamcomp/KatzS97}.

\begin{problem}[Distance selection in boxes]
Given a set $S$ of $n$ points on the two-dimensional integer grid $[2^{w}]^2$, where $w=\mathrm{polylog}(n)$, preprocess it to answer queries of the form: Given two two-dimensional boxes $B_1$ and  $B_2$ and an integer $k$, return the pair of points with the $k$th largest distance among all pairs in $(S\cap B_1)\times(S\cap B_2)$.
\end{problem}
\begin{theorem}[Distance selection in boxes]
    For any choice of a constant $0<\alpha <1$, we can construct in $\Tilde{O}(n^2)$ time with high probability a data structure for the distance selection in boxes problem using $\Tilde{O}(n^{2-\alpha/4})$ words of space which supports queries in $\Tilde{O}(n^\alpha)$ time.
\end{theorem}

Similarly, we can obtain the same bounds for the ranking problem, hence the problem of counting the number of pairs of points in $(S\cap B_1)\times(S\cap B_2)$ whose distance is at most some given integer.

For $t=3$, we can define the function $\delta$ as twice the area of the triangle defined by the three points. Note that by Pick's Theorem, the area is half-integer. Lemma~\ref{lem:points} then directly yields a data structure for the following problem.

\begin{problem}[Triangle area selection in boxes]
Given a set $S$ of $n$ points on the two-dimensional integer grid $[2^{w}]^2$, where $w=\mathrm{polylog}(n)$, preprocess it to answer queries of the form: Given three two-dimensional boxes $B_1, B_2, B_3$ and an integer $k$, return the triangle with the $k$th largest area among all those formed by three points in $(S\cap B_1)\times(S\cap B_2)\times(S\cap B_3)$.
\end{problem}
\begin{theorem}[Triangle area selection in boxes]
    For any choice of a constant $0<\alpha <1$, we can construct in $\Tilde{O}(n^3)$ time with high probability a data structure for the triangle area selection in boxes problem using $\Tilde{O}(n^{3-\alpha/4})$ words of space which supports queries in $\Tilde{O}(n^\alpha)$ time.
\end{theorem}
This problem has previously been considered by Chazelle \cite{Chazelle1987} in the ``one-shot'' setting where there are no query boxes, $k$ is fixed, and the points have real coordinates. 

Other problems of interest involve functions with divisions or radicals, for which our approach focusing on integer functions does not seem to apply at first glance. In many cases we can circumvent this problem, due to the fact that a limited precision is enough to decide which of two values is larger. One bound implied for example by the work of Burnikel et al.\ \cite{BFMS00} is the following:

\begin{lemma}[\cite{BFMS00}]\label{lem:precision}
    Let $f\colon[2^w]^d\to \mathbb{R}$ be a function which has a constant-sized expression consisting of the operators $+$,$-$,$\times$, $/$, $\sqrt{\cdot}$, $O(w)$-bit constants 
    and the arguments of $f$. Then computing $f$ up to $O(w)$ bits of precision is enough to decide if $f(p)\leq f(q)$ for any $p,q \in [2^w]^d$.
\end{lemma}

In what follows, given a function $f$ for which this lemma applies, we will use the notation $f^\bullet$ to denote the function computed to the required $O(w)$ of bits of precision, scaled to have an integer domain. Thus, $f^\bullet$ is a function $[2^w]^d \to [2^{O(w)}]$. Note that $f^\bullet$ is always computable in $\mathrm{poly}(w) = \Tilde{O}(1)$ time.

An example here is the case $t=2$, $d=2$, and $\delta (p_1,p_2) = f^\bullet(p_1,p_2)$, where $f(p_1,p_2)$ is defined as the slope of the line through $p_1,p_2$.
This yields a data structure version of the well-studied slope selection problem~\cite{DBLP:journals/siamcomp/ColeSSS89,DBLP:journals/ipl/Matousek91a,DBLP:journals/ipl/KatzS93,DBLP:journals/comgeo/BronnimannC98}.
It also applies directly to a notion from robust statistics: the \textit{Theil-Sen estimator} of a set $S$ of points in $\mathbb{R}^2$, defined as the median of the slopes of all the lines through two points of $S$. It is considered as a more robust estimator for linear regression than the classical least-square estimator, and commonly used in practice. 

\begin{problem}[Linear regression in a box]
Given a set $S$ of $n$ points on the two-dimensional integer grid $[2^{w}]^2$, where $w=\mathrm{polylog}(n)$, preprocess it to answer queries of the form: Given a two-dimensional box $B$, return the Theil-Sen estimator of the points in $S\cap B$.
\end{problem}
Applying Lemma~\ref{lem:points}, we directly obtain the following.

\begin{theorem}[Linear regression in a box]
    For any choice of a constant $0<\alpha <1$, we can construct in $\Tilde{O}(n^2)$ time with high probability a data structure for linear regression in a box using $\Tilde{O}(n^{2-\alpha/4})$ words of space which supports queries in $\Tilde{O}(n^\alpha)$ time.
\end{theorem}

Another example is the following problem, which is a data structure variant of a problem previously considered by Bespamyatnikh and Segal \cite{Bespamyatnikh2004}.
\begin{problem}[Hyperplane distance selection]
    Given $n$ points on the $d$-dimensional grid $[2^w]^d$, preprocess them to answer queries of the following form: Given $k$ and $d$ boxes $B_1,\ldots,B_d$, return the hyperplane with the $k$th largest distance to the origin, among all hyperplanes spanned by $d$ points in $B_1\times B_2 \times \cdots \times B_d$.
\end{problem}
Let $f$ be the function which maps $d$-tuples of points in $[2^w]^d$ to the distance between the origin and the hyperplane they span. Note that for any constant $d$, we can express this function in a way to which Lemma~\ref{lem:precision} applies. Thus, applying Lemma~\ref{lem:points} with $\delta = f^\bullet$ yields the following.
\begin{theorem}[Hyperplane distance selection]
    For any choice of a constant $0<\alpha <1$, we can construct in $\Tilde{O}(n^d)$ time with high probability a data structure for hyperplane distance selection using $\Tilde{O}(n^{d-\alpha/4})$ words of space which supports queries in $\Tilde{O}(n^\alpha)$ time.
\end{theorem}

As a last example, let us consider the following problem, which is the main focus of previous work by Aronov et al.\ \cite{AronovESZ23}.

\begin{problem}[Collinearity indexing with queries on a line]
Given two sets $S_1$ and $S_2$ of $n$ points on the two-dimensional integer grid $[2^{w}]^2$, where $w=\mathrm{polylog}(n)$, preprocess them to answer queries of the form: Given a point $q$ on the vertical axis, are there two points $p_1\in S_1$ and $p_2\in S_2$ such that $q$, $p_1$, and $p_2$ are collinear?
\end{problem}

Aronov et al. showed the following result, which now also follows from a straightforward application of our previous results and Lemma~\ref{lem:precision}.

\begin{theorem}[Theorem 1.1 of \cite{AronovESZ23}]
    For any choice of a constant $0<\alpha <1$, we can construct in $\Tilde{O}(n^2)$ time with high probability a data structure for collinearity indexing with queries on a line using $\Tilde{O}(n^{2-\alpha/3})$ words of space which supports queries in $\Tilde{O}(n^\alpha)$ time.
\end{theorem}
Note that as with previous applications, our results also allow us to restrict to pairs of points lying in two boxes given at query time. Using Lemma~\ref{lem:points}, we could also answer queries asking for the closest line spanned by a pair of points in $S_1\times S_2$, above or below the query point $p$, at the cost of increasing the space of our data structure to $\Tilde{O}(n^{2-\alpha/4})$.

\bibliography{implicit_search}

\end{document}